\documentclass[pra,aps,nopacs,onecolumn,twoside,superscriptaddress]{revtex4}

\usepackage{amsmath,amsfonts,amssymb,caption,color,epsfig,graphics,graphicx,latexsym,mathrsfs,revsymb,theorem,url,verbatim,epstopdf,float}

\usepackage{hyperref}

\usepackage{amsmath,amsfonts,amssymb,caption,hyperref,color,epsfig,graphics,graphicx,latexsym,mathrsfs,revsymb,theorem,url,verbatim,epstopdf,cleveref}
\hypersetup{colorlinks,linkcolor={blue},citecolor={blue},urlcolor={red}}

\newtheorem{definition}{Definition}
\newtheorem{proposition}[definition]{Proposition}
\newtheorem{lemma}[definition]{Lemma}

\newtheorem{theorem}[definition]{Theorem}
\newtheorem{corollary}[definition]{Corollary}
\newtheorem{conjecture}[definition]{Conjecture}

\newtheorem{remark}[definition]{Remark}
\newtheorem{example}[definition]{Example}
\newtheorem{question}[definition]{Question}

\def\squareforqed{\hbox{\rlap{$\sqcap$}$\sqcup$}}
\def\qed{\ifmmode\squareforqed\else{\unskip\nobreak\hfil
\penalty50\hskip1em\null\nobreak\hfil\squareforqed
\parfillskip=0pt\finalhyphendemerits=0\endgraf}\fi}
\def\endenv{\ifmmode\;\else{\unskip\nobreak\hfil
\penalty50\hskip1em\null\nobreak\hfil\;
\parfillskip=0pt\finalhyphendemerits=0\endgraf}\fi}
\newenvironment{proof}{\noindent \textbf{{Proof.~} }}{\qed}
\def\Dbar{\leavevmode\lower.6ex\hbox to 0pt
{\hskip-.23ex\accent"16\hss}D}
\makeatletter
\def\url@leostyle{%
  \@ifundefined{selectfont}{\def\UrlFont{\sf}}{\def\UrlFont{\small\ttfamily}}}
\makeatother
\def\bcj{\begin{conjecture}}
\def\ecj{\end{conjecture}}
\def\bcr{\begin{corollary}}
\def\ecr{\end{corollary}}
\def\bd{\begin{definition}}
\def\ed{\end{definition}}
\def\bea{\begin{eqnarray}}
\def\eea{\end{eqnarray}}
\def\bem{\begin{enumerate}}
\def\eem{\end{enumerate}}
\def\bex{\begin{example}}
\def\eex{\end{example}}
\def\bim{\begin{itemize}}
\def\eim{\end{itemize}}
\def\bl{\begin{lemma}}
\def\el{\end{lemma}}
\def\bma{\begin{bmatrix}}
\def\ema{\end{bmatrix}}
\def\bpf{\begin{proof}}
\def\epf{\end{proof}}
\def\bpp{\begin{proposition}}
\def\epp{\end{proposition}}
\def\bqu{\begin{question}}
\def\equ{\end{question}}
\def\br{\begin{remark}}
\def\er{\end{remark}}
\def\bt{\begin{theorem}}
\def\et{\end{theorem}}

\def\btb{\begin{tabular}}
\def\etb{\end{tabular}}

\newcommand{\nc}{\newcommand}

\def\a{\alpha}
\def\b{\beta}

\def\t{\theta}

\def\m{\mu}
\def\n{\nu}

\def\r{\rho}
\def\s{\sigma}

\def\ps{\psi}

\def\G{\Gamma}

 \nc{\bbA}{\mathbb{A}} \nc{\bbB}{\mathbb{B}} \nc{\bbC}{\mathbb{C}}
 \nc{\bbD}{\mathbb{D}} \nc{\bbE}{\mathbb{E}} \nc{\bbF}{\mathbb{F}}
 \nc{\bbG}{\mathbb{G}} \nc{\bbH}{\mathbb{H}} \nc{\bbI}{\mathbb{I}}
 \nc{\bbJ}{\mathbb{J}} \nc{\bbK}{\mathbb{K}} \nc{\bbL}{\mathbb{L}}
 \nc{\bbM}{\mathbb{M}} \nc{\bbN}{\mathbb{N}} \nc{\bbO}{\mathbb{O}}
 \nc{\bbP}{\mathbb{P}} \nc{\bbQ}{\mathbb{Q}} \nc{\bbR}{\mathbb{R}}
 \nc{\bbS}{\mathbb{S}} \nc{\bbT}{\mathbb{T}} \nc{\bbU}{\mathbb{U}}
 \nc{\bbV}{\mathbb{V}} \nc{\bbW}{\mathbb{W}} \nc{\bbX}{\mathbb{X}}
 \nc{\bbZ}{\mathbb{Z}}

 \nc{\bA}{{\bf A}} \nc{\bB}{{\bf B}} \nc{\bC}{{\bf C}}
 \nc{\bD}{{\bf D}} \nc{\bE}{{\bf E}} \nc{\bF}{{\bf F}}
 \nc{\bG}{{\bf G}} \nc{\bH}{{\bf H}} \nc{\bI}{{\bf I}}
 \nc{\bJ}{{\bf J}} \nc{\bK}{{\bf K}} \nc{\bL}{{\bf L}}
 \nc{\bM}{{\bf M}} \nc{\bN}{{\bf N}} \nc{\bO}{{\bf O}}
 \nc{\bP}{{\bf P}} \nc{\bQ}{{\bf Q}} \nc{\bR}{{\bf R}}
 \nc{\bS}{{\bf S}} \nc{\bT}{{\bf T}} \nc{\bU}{{\bf U}}
 \nc{\bV}{{\bf V}} \nc{\bW}{{\bf W}} \nc{\bX}{{\bf X}}
 \nc{\bZ}{{\bf Z}}

\nc{\cA}{{\cal A}} \nc{\cB}{{\cal B}} \nc{\cC}{{\cal C}}
\nc{\cD}{{\cal D}} \nc{\cE}{{\cal E}} \nc{\cF}{{\cal F}}
\nc{\cG}{{\cal G}} \nc{\cH}{{\cal H}} \nc{\cI}{{\cal I}}
\nc{\cJ}{{\cal J}} \nc{\cK}{{\cal K}} \nc{\cL}{{\cal L}}
\nc{\cM}{{\cal M}} \nc{\cN}{{\cal N}} \nc{\cO}{{\cal O}}
\nc{\cP}{{\cal P}} \nc{\cQ}{{\cal Q}} \nc{\cR}{{\cal R}}
\nc{\cS}{{\cal S}} \nc{\cT}{{\cal T}} \nc{\cU}{{\cal U}}
\nc{\cV}{{\cal V}} \nc{\cW}{{\cal W}} \nc{\cX}{{\cal X}}
\nc{\cZ}{{\cal Z}}

\nc{\hA}{{\hat{A}}} \nc{\hB}{{\hat{B}}} \nc{\hC}{{\hat{C}}}
\nc{\hD}{{\hat{D}}} \nc{\hE}{{\hat{E}}} \nc{\hF}{{\hat{F}}}
\nc{\hG}{{\hat{G}}} \nc{\hH}{{\hat{H}}} \nc{\hI}{{\hat{I}}}
\nc{\hJ}{{\hat{J}}} \nc{\hK}{{\hat{K}}} \nc{\hL}{{\hat{L}}}
\nc{\hM}{{\hat{M}}} \nc{\hN}{{\hat{N}}} \nc{\hO}{{\hat{O}}}
\nc{\hP}{{\hat{P}}} \nc{\hR}{{\hat{R}}} \nc{\hS}{{\hat{S}}}
\nc{\hT}{{\hat{T}}} \nc{\hU}{{\hat{U}}} \nc{\hV}{{\hat{V}}}
\nc{\hW}{{\hat{W}}} \nc{\hX}{{\hat{X}}} \nc{\hZ}{{\hat{Z}}}

\nc{\hn}{{\hat{n}}}

\def\dim{\mathop{\rm Dim}}

\def\lin{\mathop{\rm span}}

\def\max{\mathop{\rm max}}
\def\min{\mathop{\rm min}}

\def\rank{\mathop{\rm rank}}

\def\sr{\mathop{\rm sr}}

\def\tr{\mathop{\rm Tr}}

\def\dg{\dagger}

\def\ox{\otimes}

\def\ra{\rightarrow}

\def\sue{\subseteq}

\newcommand{\bra}[1]{\langle#1|}
\newcommand{\ket}[1]{|#1\rangle}

\newcommand{\ketbra}[2]{|#1\rangle\!\langle#2|}

\newcommand{\tbc}{\red{TO BE CONTINUED...}}

\newcommand{\red}{\textcolor{red}}

\def\Dbar{\leavevmode\lower.6ex\hbox to 0pt
{\hskip-.23ex\accent"16\hss}D}

\begin{document}
\title{Entanglement distillation in terms of Schmidt rank and matrix rank}

\author{Tianyi Ding}\email[]{ dingty@buaa.edu.cn}
\affiliation{LMIB and School of Mathematical Sciences, Beihang University, Beijing 100191, China}

\author{Lin Chen}\email[]{linchen@buaa.edu.cn(corresponding author)}
\affiliation{LMIB and School of Mathematical Sciences, Beihang University, Beijing 100191, China}
\affiliation{International Research Institute for Multidisciplinary Science, Beihang University, Beijing 100191, China}

\date{\today}


\pacs{03.65.Ud, 03.67.Mn}

\begin{abstract}
Entanglement distillation is a key task in quantum-information processing. In this paper, we distill non-positive-partial-transpose (NPT) bipartite states of some given Schmidt rank and matrix rank. We show that all bipartite states of Schmidt rank two are locally equivalent to classical-classical states, and all bipartite states of Schmidt rank three are 1-undistillable. Subsequently, we show that low-rank B-irreducible NPT states are distillable for large-rank reduced density operators by proving low-rank B-irreducible NPT state whose range contains a product vector is distillable. Eventually, we present an equivalent condition to distill $M\times N$ bipartite states of rank $\max\{M,N\}+1$.
\end{abstract}

\maketitle




\section{Introduction}

\label{sec:int}
Entanglement distillation is a core problem in quantum information theory \cite{PhysRevA.59.4206}. It is the procedure of asymptotically transforming mixed entangled states to pure entangled states through local operations and classical communication (LOCC). An entangled state is distillable if the transformation works. Distillable entangled states are valuable resource in performing quantum information processing tasks including dense coding \cite{PhysRevLett.69.2881}, teleportation \cite{teleportation1993}, key agreement \cite{PhysRevLett.67.661,PhysRevLett.77.2818}, computational speedup \cite{365700} and quantum error correction \cite{PhysRevA.54.3824}. In contrast, undistillable entangled states cannot be used for quantum tasks directly \cite{2012Distillability}. Unfortunately, the quantum system interacts with the environment inevitably, thus pure states are easily converted into mixed states.
It has been shown that all positive partial transpose (PPT) states are nondistillable \cite{2000Evidence}. A long-standing open problem is whether bipartite non-positive partial transpose (NPT) states can be distilled \cite{2000Evidence,PhysRevA.61.062312}. 

Recently, the distillability problem has been listed as one of the five open problems in quantum information theory \cite{horodecki2020open}. This problem is of high research value both theoretically and practically. From a mathematical point of view, the distillability problem is deeply related to the 2-decomposable map in operator algebra, the generalized Cauchy inequality conjecture, and Hilbert's 17th problem \cite{Hiroshima2010APO}. Although there have been some attempts to the distillability problem \cite{PhysRevA.68.022319,PhysRevA.78.022318,Horodecki_2009,Pankowski2007AFS}, the complete solution is unknown. The main difficulty is that it involves rapidly increasing number of parameters in the density matrices of states in high dimensions. 
The general approach to the distillability problem is starting with some simple cases to gain experience in solving the complicated cases. It has been shown that entangled states of rank two and three \cite{PhysRevA.78.022318, 2003Rank, PhysRevLett.78.574}, $2\times N$ NPT states, $M \times N$ entangled states of rank $N$ (where $M\le N$) \cite{2012Distillability,2003Rank}, and NPT states of rank at most four \cite{2016Distillability} are distillable. Further, the existence of PPT entangled states has been proven \cite{horodecki1997}, and the set of these states has a positive measure. 
The following two ways are possible to access the distillability problem. The first way is to determine whether bipartite states of certain ranks are distillable. The second way is to propose a more applicable index except matrix rank to determine distillability. For example, we research the distillability of NPT states in terms of Schmidt rank. The Schmidt rank is also used in the studies of unitary gates \cite{PhysRevA.93.042331}, mutually unbiased bases \cite{Chen2016ProductSA} and proving 0-entropy inequality in multipartite system \cite{Song2023ProofOA}.

In this paper we investigate the distillability of bipartite NPT states in terms of Schmidt rank and matrix rank. We show that all Schmidt-rank-two bipartite states are equivalent to classical-classical states in Lemma \ref{le:sr2=PPT}, and all Schmidt-rank-three bipartite states are 1-undistillable in Lemma \ref{thm:3xn=distill}. Secondly, we study the distillability of bipartite NPT states in the space $\bbC^M\ox\bbC^N$ of rank $N+1$. We propose a sufficient condition to the distillability in regard to the principal submatrix of negative determinant in Lemma \ref{le:kxk}. We prove that the 1-distillability of $M\times N$ NPT states of rank $N+1$ is equivalent to that of B-irreducibile $M\times N$ NPT states of rank $N+1$ in Lemma \ref{le:reducible}. Moreover we show that an $M\times N$ B-irreducible NPT state of rank $N+1$ whose range contains a product vector is 1-distillable in Theorem \ref{le:MxNrank(N+1)NPT=irreducible}. We prove that all bipartite B-irreducible $M\times N$ NPT states of rank $N+1$ are 1-distillable when $\min\{M,N\}>3$ in Lemma \ref{thm:MxNrank(N+1)NPT}. Further, we show that the 1-distillability of $M\times N$ NPT states of rank $N+1$ is equivalent to that of $3\times N$ NPT states of rank $N+1$. 

As far as we know, it is the first time to distill NPT states in terms of Schmidt rank. The Schmidt rank of states in a bipartite finite dimensional Hilbert space is a measure of entanglement \cite{2008On}.
For the distillability problem, the import of Schmidt rank helps make full use of local equivalence of the bipartite matrix in the two spaces respectively. The operator Schmidt rank has been found related to separability \cite{Cariello_2021}; the attainable bound of Schmidt rank of bipartite subspace has been proven by algebraic geometry techniques \cite{2008On}; the distillation of pure entangled state in bipartite system with certain Schmidt rank has been shown \cite{BISWAS2023128610}; and the extension of Schmidt rank to infinitely dimensional system is the monotone of SLOCC convertibility \cite{eps}. Moreover, many results for the distillability of low-matrix-rank states were found through constructing PPT states \cite{BDM+99}. Our results further extend the previous studies, and show the latest progress on the distillability problem.


The rest of this paper is organized as follows. In Sec. \ref{sec:pre} we introduce distillation related knowledge and mathematical notations in this paper. In Sec. \ref{subsec:sr} we introduce facts on Schmidt rank we used in this paper, and in Sec. \ref{subsec:blockform} we introduce facts on block form and irreducibility. In Sec. \ref{sec:sr} we determine the distillability of Schmidt rank two and three bipartite states respectively. 
In Sec. \ref{sec:rank} we present a necessary and sufficient condition for the distillability of low-rank entangled states.
In Sec. \ref{sec:con} we make a conclusion of this paper.


\section{preliminaries}
\label{sec:pre}

In this section, we introduce the facts and knowledge used in this paper. First, we refer to an $m\times n$ bipartite state $\r_{AB}$ as an order-$m$ block matrix with each block of order $n$, such that $\rank\r_A=m$ and $\rank\r_B=n$, where $\r_A$ and $\r_B$ are respectively the reduced density operators of system A and B of $\r_{AB}$. For example, the $k\times k$ identity matrix $I_k$ is a non-normalized state. It is also known that $\r_A$ can be worked out by partial trace over system $B$ of $\r_{AB}$, i.e., $\r_A=\tr_B\r_{AB}=\sum^n_{j=1}\bra{j}_B\r_{AB}\ket{j}_B$. Next, if $\r_{AB}$ is written as $\r_{AB}=\sum^m_{i,j=1}\ketbra{i}{j}\otimes M_{ij}$, then its partial transpose is denoted as $\r_{AB}^\Gamma=\sum^m_{i,j=1}\ketbra{j}{i}\otimes M_{ij}$. If $\r_{AB}^\Gamma$ is not positive semidefinite, then we say that $\r_{AB}$ is a non-positive-partial-transpose (NPT) state. Otherwise, we say that $\r_{AB}$ is a positive-partial-transpose (PPT) state. For example, a typical family of PPT states is the so-called separable states defined as the convex sum of product states. Third, we say that two bipartite states $\a$ and $\b$ are locally equivalent when there is an invertible product matrix $S=M\otimes N$ such that $\a=S \b S^\dg$. In this case, we shall refer to $\a\sim\b$.
One can straightforwardly show that, two locally equivalent states are at the same time NPT (and PPT, separable, respectively). Fourth, a bipartite entangled state $\r$ is $n$-distillable when there is a Schmidt-rank-two state $\ket{\ps}_{AB}$, with $A=A_1...A_n$ and $B=B_1...B_n$, such that $\bra{\ps}(\r^{\Gamma})^{\otimes n}\ket{\ps}<0$, and the $j$'th matrix $\r^\Gamma$ acts on the bipartite space $\cH_{A_jB_j}$. We say that $\r$ is distillable when $\r$ is $n$-distillable for some $n$. It is not hard to see that, two locally equivalent states are at the same distillable or not. Further, if $\r$ is PPT then the inequality $\bra{\ps}(\r^{\Gamma})^{\otimes n}\ket{\ps}<0$ fails. So every PPT state is not distillable. However, whether every NPT state is distillable remains a long-standing open problem. 

The following fact in Lemma \ref{le:2xn=NPT}(i) is from \cite{1995Purification,2000Evidence,PhysRevLett.78.574}, the fact in Lemma \ref{le:2xn=NPT}(ii) is from \cite{2003Rank,2012Distillability}, the fact in Lemma \ref{le:2xn=NPT}(iii) is from \cite{2003Rank}, and the fact in Lemma \ref{le:2xn=NPT}(iv) is from \cite{2012Distillability,2000Operational,2012Properties}. 

\begin{lemma}
\label{le:2xn=NPT}
(i) Every $2\times n$ NPT state is 1-distillable.

(ii) Every $m\times n$ NPT state of rank at most $\max\{m,n\}$ is 1-distillable.

(iii) Every $m\times n$ state of rank smaller than $\max\{m,n\}$ is NPT. 

(iv) Every $m\times n$ PPT state of rank $\max\{m,n\}$ is separable. Further, it is the convex sum of $\max\{m,n\}$ pure product states. In particular, if $m\le n$ then an $m\times n$ PPT state of rank $n$ can be written as $C^\dg C$ where $C=[C_0,...,C_{M-2},I]$ and all $C_j$'s are normal matrices which are simultaneously unitarily diagonalizable.  
\end{lemma}

In the rest of this section, we review the notions and facts on Schmidt rank, block form and irreducibility of a bipartite state. They are respectively presented in Sec. \ref{subsec:sr} and \ref{subsec:blockform}.

\subsection{Schmidt rank}
\label{subsec:sr}

We say that two matrices $A$ and $B$ are orthogonal in the sense that $\tr(A^\dg B)=0$. The following fact is from \cite{000266183600001}.

\begin{lemma}
\label{le:HermitianDECOMP}
Suppose $\r$ is an $m\times n$ state. Then $\r=\sum_j A_j \otimes B_j$, where $A_j$'s are nonzero orthogonal order-$m$ Hermitian matrices, and $B_j$'s are nonzero orthogonal order-$n$ Hermitian matrices.    
\end{lemma}

Let $\bbM_{a,b}(\bbC)$ denote the set of $a\times b$ matrices. It is known that the \textit{Schmidt rank} of a bipartite matrix $M$ in the space $\bbM_{a,b}(\bbC)\otimes \bbM_{c,d}(\bbC)$ is the minimum integer $k$ admitting the decomposition $M=\sum^k_{j=1} A_j \otimes B_j$ with $A_j\in \bbM_{a,b}(\bbC)$ and $B_j\in \bbM_{c,d}(\bbC)$.
If $k$ is exactly the Schmidt rank of $M$ then the decomposition is referred to as a Schmidt decomposition of $M$, and we denote $k=\sr(M)$. It is not hard to verify that the decomposition is a Schmidt decomposition of $M$ if and only if $A_j$'s are linearly independent, and 
$B_j$'s are also linearly independent. For example, 
the decomposition in Lemma \ref{le:HermitianDECOMP} is a Schmidt decomposition of $\rho$, and the number of matrices $A_j$'s in the lemma is the Schmidt rank of $\r$.
For convenience, we shall refer to \textit{the space A (or B)} of the bipartite matrix $M$ as the matrix subspace spanned by $A_j$'s (or $B_j$'s) in the Schmidt decomposition $M=\sum^k_{j=1} A_j \otimes B_j$. Using these notions, one can straightforwardly verify the following claim. 
\begin{lemma}
 \label{le:constantSPACEA+B} 
(i) Every bipartite matrix $M$ has constant space $A$ and $B$. They have the same dimension. 

(ii) If $F_1,..,F_s$ are linearly independent matrices in the space $A$ of $M$, then we can find a Schmidt decomposition of $M$ such that $M=\sum^s_{j=1}F_j\otimes B_j+\sum^{\sr(M) }_{j=s+1}A_j\otimes B_j$. 

(iii) If $M$ is a bipartite state and $F_1,..,F_s$ are linearly independent Hermitian matrices in the space $A$ of $M$, then we can find a Schmidt decomposition of $M$ such that $M=\sum^s_{j=1}F_j\otimes B_j+\sum^{\sr(M) }_{j=s+1}A_j\otimes B_j$, where all $A_j$ and $B_j$ are Hermitian matrices. 

(iv) If $M$ is a bipartite state then the reduced density operator of system $A$ is in the space $A$ of $M$.
\end{lemma}

The following result relatd to Schmidt rank is from \cite{2008On}.
 \bl
\label{le:maximumdim}
(i) The maximum dimension of a subspace $\cS \sue \cH_A \ox \cH_B=\bbC^M\otimes\bbC^N$
is given by
$(M-k+1)(N-k+1)$, when
$\sr(\ket{v})\ge k$ for all nonzero $\ket{v} \in \cS$.

(ii) If a subspace $\cV \sue \cH_A \ox \cH_B=\bbC^M\otimes\bbC^N$ has dimension $\dim\cV  (M-1)(N-1)+1$, then $\cV$ contains infinitely many product vectors.
 \el
 It can be deduced that if $\cU$ is a subspace of $\cH_A \ox\cH_B=\bbC^M\otimes\bbC^N$ and $\dim \cU > (M-k+1)(N-k+1)$, then there exists $\ket{u}\in \cU$, such that $\sr(\ket{u})<k$.
For example, the lemma says that every $M\times N$ subspace of dimension at least $(M-1)(N-1)+1$ has a product vector. Further, every $M\times N$ subspace of dimension at least $(M-2)(N-2)+1$ has a bipartite vector of Schmidt rank two.

We shall apply the facts of this subsection in Sec. \ref{sec:sr}. 
 
\subsection{block form and irreducible state}
\label{subsec:blockform}

We will denote by $\{\ket{i}_A:i=0,\ldots,M-1\}$ and
$\{\ket{j}_B:j=0,\ldots,N-1\}$ as the orthonormal bases of $\cH_A$ and $\cH_B$,
respectively. Let $\cH=\cH_A\ox\cH_B$. The subscripts A and B will be often omitted. Any
state $\r$ of rank $r$ can be represented as
 \bea
 \label{ea:MxN-State}
\r=\sum_{i,j=0}^{M-1} \ketbra{i}{j}\ox C_i^\dag C_j,
 \eea
where the $C_i$ are $R\times N$ matrices and $R$ is an arbitrary
integer $\ge r$. In particular, one can take $R=r$. We shall often
consider $\r$ as a block matrix $\r=C^\dag C=[C_i^\dag C_j]$, where
$C=[C_0~C_1~\cdots~C_{M-1}]$ is an $R\times MN$ matrix. Thus
$C_i^\dag C_j$ is the matrix of the linear operator $\bra{i}_A \r
\ket{j}_A$ acting on $\cH_B$. For the reduced density matrices, we
have the formulae
 \bea \label{RedStates}
\r_B=\sum_{i=0}^{M-1} C_i^\dag C_i; \quad \r_A=[\tr C_i^\dag C_j],
\quad i,j=0,\ldots,M-1.
 \eea

One can verify that the range of $\r$ is the  space spanned by the column vectors of
the matrix $C^\dag$ and that
 \bea \label{JezgroRo}
\ker\r=\left\{\sum_{i=0}^{M-1}\ket{i}\ox\ket{y_i}:
\sum_{i=0}^{M-1}C_i\ket{y_i}=0\right\}.
 \eea
In particular, if $C_i\ket{j}=0$ for some $i$ and $j$ then
$\ket{i,j}\in\ker\r$.

Next we introduce the concept of irreducibility for bipartite states.

 \bd
\label{df:red} 
We say that a linear operator
$\r:\cH\to\cH$  is an {\em A-direct sum} of linear operators
$\r_1:\cH\to\cH$ and $\r_2:\cH\to\cH$, and we write
$\r=\r_1\oplus_A\r_2$, if
$\cR(\r_A)=\cR((\r_1)_A)\oplus\cR((\r_2)_A)$. A bipartite state $\r$
is {\em A-reducible} if it is an A-direct sum of two states;
otherwise $\r$ is {\em A-irreducible}. One defines similarly the
{\em B-direct sum} $\r=\r_1\oplus_B\r_2$, the {\em B-reducible} and
the {\em B-irreducible} states. We say that a state $\r$ is {\em
reducible} if it is either A or B-reducible. We say that $\r$ is
{\em irreducible} if it is not reducible. We write
$\r=\r_1\oplus\r_2$ if $\r=\r_1\oplus_A\r_2$ and
$\r=\r_1\oplus_B\r_2$, and in that case we say that $\r$ is a {\em
direct sum} of $\r_1$ and $\r_2$.
 \ed

The following fact is from \cite{2012Distillability}.
\begin{lemma}
\label{le:irreducible}
    Let $\r$ be an B-irreducible $M\times N$ state such that $\cH'_A\ox \ket{b}\subseteq \ker\r$ for an $(M-1)$-dimensional subspace $\cH'_A\subseteq\cH_A$ and some state $\ket{b}\in\cH_B$. Then $\r$ is 1-distillable.
\end{lemma}

We shall apply the facts of this subsection in Sec. \ref{sec:rank}.

\section{Distillability in terms of Schmidt rank}
\label{sec:sr}

We recall that the classical-classical state is a separable state having no quantum correlation. That is, the state has zero quantum discord \cite{LHenderson_2001} and is mathematically a diagonal state. The first assertion of the following lemma has been proven in \cite{2019Separability}. We give a novel proof for the other claims of lemma. 

\begin{lemma}
\label{le:sr2=PPT}
    Every bipartite state of Schmidt rank two is separable, and thus cannot be distillable. Besides, the state is locally equivalent to a classical-classical state.
\end{lemma}
\begin{proof}
Suppose $\r_{AB}$ is a bipartite state of Schmidt rank two on the bipartite space $\cH_{AB}$. Lemma \ref{le:HermitianDECOMP} says that, $\r_{AB}$ can be written as 
\begin{eqnarray}
\label{eq:rho=A1}
\r_{AB}=A_1\otimes B_1+A_2 
\otimes B_2,
\end{eqnarray}
where $A_j$ and $B_j$ are all Hermitian matrices. By tracing out system $B$, we obtain that 
$
\r_A=A_1 (\tr B_1) +A_2 
(\tr B_2).
$
Because $\r_A$ is an invertible matrix, one of the two real numbers $\tr B_1$ and $\tr B_2$ is nonzero. Without loss of generality, we may assume that $\tr B_1\ne0$. Hence $A_1$ is the real linear combination of $\r_A$ and $A_2$. Applying the combination to \eqref{eq:rho=A1}, we obtain that $\r_{AB}=\r_A\otimes B_1'+A_2\otimes B_2'$. Because $\r_A$ is invertible, we can find an invertible matrix $S$ such that $S\r_A S^\dg$ is the identity matrix, and $SA_2 S^\dg $ is a real diagonal matrix. Applying the above argument to $B_1$ and $B_2$ in 
\eqref{eq:rho=A1}, we can find an invertible matrix $W$ such that $W\r_B W^\dg$ is the identity matrix, and $WB_2 W^\dg $ is a real diagonal matrix. As a result, the state $(S\otimes W)\r_{AB}(S^\dg \otimes W^\dg)$ is diagonal, which is a classical-classical state having no quantum correlation. Because the separability is invariant under local equivalence, we know that $\r_{AB}$ is also separable, and thus cannot be distillable. 
\end{proof}

Next, we proceed with bipartite states of Schmidt rank three. It turns out that such states are 1-undistillable or  undistillable.

\begin{theorem}
\label{thm:3xn=distill}
(i) If the space A of a bipartite state of Schmidt rank three has an Hermitian matrix of rank one, then the state is PPT and not distillable.
    
    (ii) For every $\min\{m,n\}>2$, every $m\times n$ NPT state of Schmidt rank three is 1-undistillable.

    (iii) Every $2\times n$ state of Schmidt rank three is separable. 
\end{theorem}
\begin{proof}
(i) Suppose $\r_{AB}$ is a bipartite state of Schmidt rank three on the bipartite space $\cH_{AB}=\bbC^m\otimes\bbC^n$. Using Lemmas \ref{le:constantSPACEA+B} (iii), we can assume that 
\begin{eqnarray}
\label{eq:rho=A123}
\r_{AB}=A_1\otimes B_1+A_2 
\otimes B_2+A_3\otimes B_3,
\end{eqnarray}
where $A_j$ and $B_j$ are all Hermitian matrices, $A_1=\r_A$ and $\rank A_2=1$. Up to the local equivalence, we can assume that $A_1=I_m$ and $A_2=1\oplus 0_{m-1}$. Because $A_3$ is Hermitian, one can find a unitary gate $U=1\oplus V$ such that $UA_3U^\dg$ is a real symmetric tridiagonal matrix. Hence, the partial transpose of $(U\otimes I)\r_{AB}(U^\dg \otimes I)$ is still positive semidefinite. So $\r_{AB}$ is PPT and not distillable.

(ii) For every bipartite state $\r_{AB}$ of Schmidt rank three, we can 
still apply  \eqref{eq:rho=A123}, in which $A_1=I_m$, and $A_2,A_3$ are both Hermitian matrices. 

We disprove the assertion. Suppose for some $\min\{m,n\}>2$, there is an $m\times n$ NPT state $\r_{AB}$ of Schmidt rank three is 1-distillable. Hence, there is a rank-two matrix $M$ such that the projected state $\s_{AB}=(M\otimes I)\rho_{AB}(M^\dg \otimes I)$ is entangled and NPT. Let $U$ be an order-$m$ invertible matrix such that the nonzero entries of $UM$ are all in the top two rows of $UM$. So we can refer $UM$ to the first two rows of it and assume that $\s_{AB}$ is a $2\times2$ block matrix of each block of order $n$. It follows from \eqref{eq:rho=A123} that $\s_{AB}=\sum^3_{j=1}A_j'\otimes B_j$, where $A_j'$'s are all $2\times2$ Hermitian matrices, and $A_1'$ is invertible. So we can find a $2\times2$ invertible matrix $S$ such that $SA_j' S^\dg$'s are all real symmetric matrices. As a result, the state $(S\otimes I)\s_{AB}(S^\dg \otimes I)$ is a $2\times n$ PPT state. It is a contradiction with the fact that $\s_{AB}$ is NPT. We have proven the assertion.

(iii) It follows from the proof of (ii) that every $2\times n$ state $\r_{AB}$ of Schmidt rank three is PPT. In particular, up to a local equivalence we have $\r_{AB}^\Gamma=\r_{AB}$. It has been proven that such $\r_{AB}$ is separable \cite{2000Separability}.
\end{proof}

So far we have investigated the distillability of NPT states in terms of Schmidt rank two and three. Nevertheless, we have found that it is not easy to apply the existing results to NPT states of Schmidt rank four. To further distill more entangled states, we shall investigate their rank in the next section. Such states may have various Schmidt rank, as we will in see in the next section.

\section{Distillability in terms of matrix rank}
\label{sec:rank}

In this section, we distill some types of bipartite NPT states in terms of rank. In Lemma \ref{le:kxk}, we provide a sufficient condition to determine the distillability which paves the way for proving Lemma \ref{le:MxNrank(N+1)NPT=irreducible}. Next, in Lemma \ref{le:reducible}, we show that the 1-distillablility of all bipartite NPT states is equivalent to that of all B-irreducible bipartite NPT states. Further, the 1-distillability of all $M\times N$ NPT states of rank $\max\{M,N\}+1$ is equivalent to that of all B-irreducible $M\times N$ NPT states of rank $\max\{M,N\}+1$. Lemma \ref{le:MxNrank(N+1)NPT=irreducible} and Theorem \ref{thm:MxNrank(N+1)NPT} are crucial work in this section. Lemma \ref{le:MxNrank(N+1)NPT=irreducible} shows the relation between the  distillability and range of bipartite states containing a product vector. This lays the foundation to show that all $M\times N$ B-irreducible NPT states of rank $N+1$ are distillable in Theorem \ref{thm:MxNrank(N+1)NPT}. Eventually, we derive that the 1-distillability of every $M\times N$ NPT state of rank $\max\{M,N\}+1$ is equivalent to that of every $3\times N$ B-irreducible NPT state of rank $\max\{M,N\}+1$ in Lemma \ref{le:equivalence}.

The following observation extends Lemma 2 of \cite{2016Distillability}.
\begin{lemma}
\label{le:kxk}
Let $\r$ be a bipartite state such that $\r^\G$ has a principal $k\times k$ submatrix of negative determinant. If the diagonal elements of the submatrix are from two blocks of $\r=\sum_{i,j}\ketbra{i}{j}\otimes \s_{ij}$, then $\r$ is distillable.
\end{lemma}

\begin{proof}
Let $\r=\sum_{i,j}\ketbra{i}{j}\otimes \s_{ij}$, then $\r^\G=\sum_{i,j}\ketbra{j}{i}\otimes \s_{ij}$. Denote the submatrix as $[s_{ij}]$. From $\s_{ii}\geq0$ we know that the diagonal entries of $[s_{ij}]$ belong to different blocks, say $\s_{ll},\s_{mm}$. Let P be the orthogonal projector onto
the $2$-dimensional subspace of $\cH_A$ spanned by $\{\ket{l},\ket{m}\}$, that is $P=\ketbra{l}{l}+\ketbra{m}{m}$. Then the projected state $(P\otimes I)\r(P^\dg\otimes I)$ has a principal $k\times k$ submatrix of negative determinant, thus it is an $2\times n$ NPT state. Using Lemma \ref{le:2xn=NPT}, we obtain that $\r$ is distillable.
\end{proof}

As an application of Definition \ref{df:red}, we present the following observation.

\begin{lemma}
\label{le:reducible}
(i) All bipartite NPT states are 1-distillable if and only if all B-irreducible bipartite NPT states are 1-distillable. 

(ii) Let $M\le N$. All $M\times N$ NPT states of rank $N+1$ are 1-distillable if and only if all B-irreducible  $M\times N$ NPT states of rank $N+1$ are 1-distillable. 
\end{lemma}
\begin{proof}
(i) It suffices to prove the "if" part. We assume that all B-irreducible bipartite NPT states are 1-distillable.
We take an arbitrary B-reducible bipartite NPT state $\r$ which can be decomposed into $\r=\r_1\oplus_B\r_2$. If $\r_1$ or $\r_2$ is still B-reducible, we continue to decompose it into the B-direct sum of two linear operaters, until $\r=(\oplus_{i})_B\r_i$ where $\r_i$'s are B-irreducible. 
Since $\r$ is an NPT state, for some $j$ we have that $\r_j$ is an NPT state. We can take a projector to transform $\r$ to $\r_j$, from our hypothesis $\r_j$ is 1-distillable and thus $\r$ is 1-distillable.

(ii) It still suffices to prove the "if" part. We assume that all B-irreducible $M\times N$ NPT states of rank $N+1$ are 1-distillable. For an arbitrary $M\times N$ NPT state $\r$ of rank $N+1$, we decompose it into the sum of B-irreducible states $\r=(\oplus_i)_B \r_i$. If there exists some $m_j\times n_j$  state $\r_j$ such that $\rank(\r_j)<n_j$, by Lemma \ref{le:2xn=NPT}(iii) $\r_j$ is an NPT state. We can transform $\r$ to $\r_j$, and by Lemma \ref{le:2xn=NPT}(ii) we know that $\r$ is 1-distillable. Otherwise, for any $i$, $\r_i$ is an $m_i\times n_i$ state with $\rank(\r_i)\geq n_i$. Since
\begin{eqnarray}
\label{eq:Rrhok_subseteq}
    \cR(\r_k)\subseteq\cR((\r_k)_A)\ox\cR((\r_k)_B)
\end{eqnarray}
and 
\begin{eqnarray}
    \cR(\r_l)\subseteq\cR((\r_l)_A)\ox\cR((\r_l)_B)
\end{eqnarray}
for any $k\neq l$, where $\cR((\r_k)_B) \cap\cR((\r_l)_B)=\{0\}$. Then $\rank(\r)=\sum_i \rank(\r_i)$.
In terms of $\rank((\r_i)_B)= n_i\leq\rank(\r_i)$,
we find that 
\begin{eqnarray}
    N=\rank(\r_B)=\sum_i \rank((\r_i)_B)\leq \sum_i \rank(\r_i)=\rank(\r)=N+1.
\end{eqnarray}
Therefore, for any $i$ either 
\begin{eqnarray}
\label{eq:rank(riB)=rank(ri)}
    \rank((\r_i)_B)=\rank(\r_i)
\end{eqnarray} or 
\begin{eqnarray}
\label{eq:rank(riB)+1=rank(ri)}
    \rank((\r_i)_B)+1=\rank(\r_i)
\end{eqnarray}
stands. Because $\r$ is an NPT state, we obtain that $\r_j$ is an NPT state for some $j$. If \eqref{eq:rank(riB)=rank(ri)} stands, $\r_j$ is an $m_j\times n_j$ NPT state with rank $n_j$, by Lemma \ref{le:2xn=NPT}(ii) $\r_j$ is 1-distillable. If \eqref{eq:rank(riB)+1=rank(ri)} stands, $\r_j$ is an $m_j\times n_j$ NPT state with rank $n_j+1$, when $m_j>n_j$ by Lemma \ref{le:2xn=NPT}(ii) $\r_j$ is 1-distillable; when $m_j\leq n_j$ by the hypothesis $\r_j$ is 1-distillable. Under a projector $P$, $(I\ox P^\dg)\r(I\ox P)=\r_j$, so $\r$ is 1-distillable.
\end{proof}

In the following, we investigate a special case of distilling B-irreducible NPT states.

 \bl
\label{le:MxNrank(N+1)NPT=irreducible}
If $\r$ is an $M\times N$ B-irreducible NPT state of rank $N+1$
such that the range of $\r$ contains at least one product vector,
then $\r$ is 1-distillable.
 \el
\begin{proof}
By using Eq. (\ref{ea:MxN-State}), we have $\r=C^\dag C$ where
$C=[C_0~C_1~\cdots~C_{M-1}]$ and the blocks $C_i$'s
are $(N+1)\times N$ matrices. 
The blocks are linearly independent because $\rank\r_A=M$.

We disprove the assertion. Assume that some $\r$ is not 1-distillable. Using Lemma \ref{le:2xn=NPT}, we have $2<M\le N$. Since $\cR(\r)$ contains at least one product vector, we may
assume that the first row of $C_i$ is 0 for $i>0$.
For the state $\s:=[C_1~\cdots~C_{M-1}]^\dag[C_1~\cdots~C_{M-1}]$,
we have $\rank\s \le N$ and $\s_B=\sum_{i>0} C_i^\dag C_i$. If $\s_B\ket{b}=0$ for some
$\ket{b}\ne0$, then $C_i\ket{b}=0$ for $i>0$ and so
$\ket{0}^\perp\ox\ket{b}\subseteq\ker\r$. So $\r$ is
1-distillable by Lemma \ref{le:irreducible}, which is a contradiction with the assumption. Thus we may assume that $\rank\s_B=N$. Next Lemma \ref{le:2xn=NPT} (iii) implies that 
$\rank\s=N$. Lemma \ref{le:2xn=NPT} (iv) implies that $\s$ is PPT, and $\s$ is the convex sum of exactly $N$ pure product states. 

By dropping the first row of $C_i$, we obtain the $N\times N$ matrix
$C'_i$, $i=0,1,\ldots,M-1$. 
Up to the local equivalence $I\otimes V$ for some invertible $V$, we may assume that $C'_1=I_N$,
 \bea \label{jed:C'-blokovi}
 C'_i=\lambda_{i1} I_{l_1}\oplus\cdots\oplus \lambda_{ik} I_{l_k},
 \quad i>0; \quad l_1 + \cdots + l_k=N,
 \eea
and that whenever $r\ne s$ there exists an $i>1$ such that
$\lambda_{ir}\ne\lambda_{is}$. (Note that all $\lambda_{1r}=1$.) Since the $C_i$
are linearly independent, each set $\{\lambda_{ir}:r=1,\ldots,k\}$,
$i>1$, must have at least two elements, namely $k>1$.
Let $C'_0=[A_{ij}]_{i,j=1}^k$ be the partition with the $A_{ii}$
square of order $l_i$. If some $A_{rs}\ne0$, $r\ne s$, then $\r$ is 1-distillable by Lemma \ref{le:kxk} and projecting $\r\ra P^\dg\r P$ where $P=(\ket{0}(\bra{0}+x\bra{1})+\ket{j}\bra{1})\otimes I$ for some $x$ and $j>0$. It is a contradiction with the assumption. Hence, we may assume that $C'_0=B_1\oplus\cdots\oplus B_k$ with
$B_i=A_{ii}$ square of order $l_i$ and upper triangular.

Suppose the first row of $C_0$ consists of the vectors
$w_1,\ldots,w_k$ of lengths $l_1,\ldots,l_k$, respectively. Let
$\m_i$ and $\nu_i$ be the first entries of $w_i$ and $B_i$,
respectively. If some $\m_i$ is 0, say $\m_1=0$, then by subtracting from $C_i$,
$i\ne1$, a suitable scalar multiple of $C_1$, we may assume that the
first columns of the $C_i$, $i\ne1$, are 0.
Since $\r$ is B-irreducible,
the second row of $C_0$ is not zero.
It follows that the state $[C_0~C_1]^\dg [C_0~C_1]$ is 1-distillable by Lemma \ref{le:irreducible},
and so is $\r$. It is a contradiction with the assumption.
Thus, we may assume that all $\m_i\ne0$. 

Because $\r$ is NPT, at least one block of $B_j$'s, say $B_1$ is not a diagonal matrix. We can find one block of $C_j$'s with $j>1$, say $C_2$ satisfying $\lambda_{21}\ne \lambda_{2k}$. Now the projected state $K^\dg K$ where $K=[C_0-\n_k C_1, C_2-\lambda_{2k}C_1]$ has to be PPT. We can further project the state $K^\dg K$ onto $L^\dg L$ where $L=[B_1-\n_k I_{l_1},(\lambda_{21}-\lambda_{2k})I_{l_1}]$. Note that this is a $2\times l_1$ PPT state of rank $l_1$. Using Lemma \ref{le:2xn=NPT} (iii), we know that $B_1-\n_k I_{l1}$ is normal and thus the upper triangular matrix $B_1$ is diagonal. It is a contradiction with the fact that $B_1$ is not diagonal. So the assumption that $\r$ is not 1-distillable fails, and we have completed the proof.
\end{proof}

Now we are in a position to present the main result of this section.

\begin{theorem}
\label{thm:MxNrank(N+1)NPT}
$(M,N>3)$ If $\r$ is an $M\times N$ B-irreducible NPT state of rank $N+1$, then $\r$ is 1-distillable.
\end{theorem}
\begin{proof}
We disprove the claim. Assume that some $\r$ is not 1-distillable.
By using Lemma \ref{le:MxNrank(N+1)NPT=irreducible}, we may assume that the range of $\r$ has no product vector. By Lemma \ref{le:2xn=NPT} (ii), we may assume that $M\le N$. 
For any $\ket{a}\in\cH_A$ let $r_a$ be the rank of the linear
operator $\bra{a}\r\ket{a}$. Since $M>3$, we have $\dim\ker\r=MN-N-1>(M-1)(N-1)+1$. So by Lemma \ref{le:maximumdim} (ii) $\ker\r$ contains infinitely many product vectors. If $\ket{a,b}\in\ker\r$ is a product vector then
$\bra{a}\r\ket{a}$ kills the vector $\ket{b}$, and so $r_a<N$. Let
$R$ be the maximum of $r_a$ taken over all $\ket{a}\in\cH_A$ such
that $r_a<N$. Thus $R<N$. Without any loss of generality we may
assume that $\bra{0}_A\r\ket{0}_A$ has rank $R$.

We write $\r$ as in Eq. (\ref{ea:MxN-State}), i.e., $\r=C^\dag C$
where $C=[C_0~\cdots~C_{M-1}]$ and the blocks $C_i$ are $(N+1)\times
N$ matrices. As $C_0^\dag C_0=\bra{0}_A\r\ket{0}_A$, $\rank C_0=R$ and we may
assume that
 \begin{equation}
C_0=\left[\begin{array}{cc}
    I_R & 0
 \\ 0 & 0 \end{array}\right]; \quad
C_i=\left[\begin{array}{cc}
    C_{i0} & C_{i1}
 \\ C_{i2} & C_{i3} \end{array}\right],~ i>0,
 \end{equation}
where the $C_{i0}$ are $R\times R$ matrices. Because $\r$ is not 1-distillable, Lemma \ref{le:kxk} implies that all $C_{i1}=0$. The state 
\begin{eqnarray}
\label{eq:sigma'}    
\s={C'}^\dag C', \end{eqnarray}
where $C'=[C_{1,3}~\cdots~C_{M-1,3}]$,
acts on the bipartite subspace $\ket{0}^\perp\otimes\lin\{\ket{R},...,\ket{N-1}\}$. If $\s_B$ is not of full rank, say $\s_B\ket{b}=0$ for some nonzero $\ket{b}$, then $C_j\bma 0_R \\ \ket{b} \ema=0$. It implies that $\r_B$ is singular, which is a contradiction with the fact that $\r_B>0$. So $\s_B>0$, and $\rank\s_B=N-R$. 
Next, the assumption that $\r$ is not 1-distillable implies that $\rank\s\ge \rank\s_B=N-R$ by Lemma \ref{le:2xn=NPT}.
Recall that the blocks $C_i$ are $(N+1)\times
N$ matrices, and thus $\rank\s=\rank C'\le N-R+1$. We have two cases namely (i) $\rank\s=N-R$, and (ii) $\rank\s=N-R+1$.

(i) Suppose $\rank\s=N-R$. Because $\r$ is not 1-distillable and $ \rank\s_B=N-R$, we see that $\s$ is PPT state by Lemma \ref{le:2xn=NPT} (ii). And Lemma \ref{le:2xn=NPT} (iv) shows that $\s$ is the convex sum of exactly $N-R$ pure product states. Up to local equivalence, we may assume that every $C_{i3}$ is diagonal with the last zero row for $i>0$. By the same reason, we may assume that the first entry of $C_{13}$ is nonzero and that of $C_{i3}$ for $i>1$ is zero. Hence $\ket{1}^{\perp}\otimes\ket{R}\subseteq\ker\r$. So $\r$ is 1-distillable by Lemma \ref{le:irreducible}. It is a contradiction with the assumption, and we have completed the proof in case (i). 

(ii) Suppose $\rank\s=N-R+1$. The fact $\rank\s_B=N-R$ and Eq. \eqref{eq:sigma'} imply that 
\begin{eqnarray}
\label{eq:1<m<M}    
1<m:=\rank\s_A<M.
\end{eqnarray}
Up to local equivalence, we may assume that the blocks $C_{i3}=0$ for $i>m$, and the blocks $C_{13},...,C_{m3}$ are linearly independent. If a nontrivial linear combination of the blocks $C_{13},...,C_{m3}$ is of rank $r_1<N-R$, then there is a vector $\ket{a'}\in\cH_A$ such that $\bra{a'}\r\ket{a'}$
has rank $R+r_1\in(R,N)$. It is a contradiction with the definition of $R$, namely $R$ is the maximum of $r_A=\rank\bra{a}\r\ket{a}$ taken over all $\ket{a}\in\cH_A$ such
that $r_a<N$. Hence, any nontrivial linear combination of the blocks $C_{13},...,C_{m3}$ is of full rank $N-R$. Using Eq.\eqref{eq:sigma'}, we obtain that $\s$ is an $m\times (N-R)$ state of rank $N-R+1$, and the kernel of $\s$ has no product vector. It occurs only if $m=2$ due to \eqref{eq:1<m<M} and Lemma \ref{le:maximumdim}. Hence $C_{i3}=0$ for $i>2$. Because $\r$ is not 1-distillable, we have $C^{\dg}_{i2}C_{13}=C^{\dg}_{i2}C_{23}=0$. Hence, $C_{i2}=0$ for $i>2$. 

For convenience, we summary so far the facts over $\r$, i.e.,
\begin{eqnarray}
\label{eq:rho=summary}
\r=[C_0,C_1,...,C_{M-1}]^\dg [C_0,C_1,...,C_{M-1}],
\end{eqnarray} 
where
 \begin{equation}
\label{eq:C0=} 
C_0=\left[\begin{array}{cc}
    I_R & 0
 \\ 0 & 0 \end{array}\right], \quad
C_1=\left[\begin{array}{cc}
    C_{10} & 0
 \\ C_{12} & C_{13} \end{array}\right],
\quad
C_2=\left[\begin{array}{cc}
    C_{20} & 0
 \\ C_{22} & C_{23} \end{array}\right],
\quad
C_i=\left[\begin{array}{cc}
    C_{i0} & 0
 \\ 0 & 0 \end{array}\right],~ M>i>2, 
 \end{equation}
where the blocks $C_{i0}$ are $R\times R$ matrices, $\rank [C_{13} \quad C_{23}]=N-R+1$, the blocks $C_{13}$ and $C_{23}$ are linearly independent $(N-R+1)\times(N-R)$ blocks, any nontrivial linear combination of $C_{13}$ and $C_{23}$ is of full rank $N-R$.

We can project the state $\r$ in \eqref{eq:rho=summary} onto $\r_1=[I_R,C_{30},...,C_{M-1,0}]^\dg [I_R,C_{30},...,C_{M-1,0}]$. One may verify that $\r_1$ and $(\r_1)_B$
both rank $R$. Because $\r$ is not 1-distillable, Lemma \ref{le:2xn=NPT} shows that $\r_1$ is PPT and the convex sum of exactly $R$ pure product states. Up to local equivalence, we may assume that every $R\times R$ matrix $C_{i0}$ in \eqref{eq:C0=} for $i>2$ is diagonal. By simultaneously permuting their diagonal
entries (if necessary) we may assume that
 \bea
 C_{i0}=\lambda_{i1} I_{l_1}\oplus\cdots\oplus \lambda_{ik} I_{l_k},
 \quad i>2; \quad l_1 + \cdots + l_k=R,
 \eea
and that whenever $r\ne s$ there exists an $i$ such that
$\lambda_{ir}\ne\lambda_{is}$. Recall that $\rank\r_A=M>3$, we have $k>1$. Using Lemma \ref{le:kxk} one can show that the matrices $C_{10}$ and $C_{20}$ are direct sums
 \bea
 \label{eq:C10}
C_{10}= E_{1}\oplus\cdots\oplus E_{k},\quad C_{20}=
F_{1}\oplus\cdots\oplus F_{k},
 \eea
where $E_{i}$ and $F_{i}$ are square blocks of size $l_i$, and we
may assume the $E_{i}$ are lower triangular.
Let us write
 \bea
C_{i2}=\left[
          \begin{array}{c}
            C_{i21} \\
            C_{i22} \\
          \end{array}
        \right],\quad
C_{i3}=\left[
          \begin{array}{c}
            C_{i31} \\
            C_{i32} \\
          \end{array}
        \right],\quad i=1,2;
 \eea
where $C_{i22}$ and $C_{i32}$ are row-vectors. By multiplying $[C_0,C_1,...,C_{M-1}]$ of \eqref{eq:rho=summary} on
the left hand side by a unitary matrix $I_R\oplus U$, we may assume
that $C_{132}=0$. Recall that $\rank C_{13}=N-R$ below \eqref{eq:C0=}. So the square matrix $C_{131}$ is invertible, and we may assume that $C_{121}=0$ up to local equivalence. We split the row-vector $C_{122}$ into $k$ pieces
$w_1,\ldots,w_k$ of lengths $l_1,\ldots,l_k$, respectively. To
summarize, the matrices $C_j$, $j>0$ in \eqref{eq:rho=summary} have the form
 \begin{eqnarray}
&& 
C_0=\left[\begin{array}{cc}
    I_R & 0
 \\ 0 & 0 
 \\ 0 & 0 \end{array}\right],
\quad
C_1=\left[\begin{array}{cc}
    \left[\begin{array}{ccc}
       E_1 & & \\
       & \ddots & \\
       & & E_k \\
      \end{array}\right]
    & 0
 \\ 0 & C_{131}
 \\
 \left[
 \begin{array}{ccc} w_1,\ldots,w_k \end{array} \right]
    & 0
 \end{array}\right],
 \quad
C_2=\left[\begin{array}{cc}
     \left[\begin{array}{ccc}
       F_1 & & \\
       & \ddots & \\
       & & F_k \\
      \end{array}\right]
    & 0
 \\ C_{221} & C_{231}
 \\ C_{222} & C_{232} \end{array}\right],
 \notag\\ &&
C_j=\left[\begin{array}{cc}
    \left[\begin{array}{ccc}
       \lambda_{j1} I_{l_1} & & \\
       & \ddots & \\
       & & \lambda_{jk} I_{l_k} \\
      \end{array}\right]
    & 0
 \\ 0 & 0
 \\ 0 & 0 \end{array}\right],\quad M>j>2,
 \quad
 k>1.
\label{eq:Cj} 
 \end{eqnarray} 
Recall that the hypothesis says that $\cR(\r)$ has no product vector. Eq. \eqref{eq:rho=summary} implies that in \eqref{eq:Cj}, each $l_i>1$ and at least one $w_i\ne0$.
As we can simultaneously permute the first $k$ diagonal blocks of
the matrices $C_j$, we may assume that $w_1\ne0$. Let
$w_1=(a_1,\ldots,a_n,0,\ldots,0)$ where $a_j\ne0$ and let us
partition
 \bea
E_1= \left[\begin{array}{cc}
             E_{10} & E_{11} \\
             E_{12} & E_{13} \\
           \end{array}\right],
 \eea
where $E_{10}$ is of size $n\times n$. Recal the asumption below \eqref{eq:C10}, $E_1$ is a lower triangular matrix, thus $E_{11}=0$. By the same reason and by using Lemma \ref{le:2xn=NPT} (iv),
we may assume that the state
$[I_{l_1-n}~E_{13}]^\dg[I_{l_1-n}~E_{13}]$ is PPT and so the matrix
$E_{13}$ must be normal. Since $E_{13}$ is also lower triangular below \eqref{eq:C10}, it
must be a diagonal matrix.

By adding a suitable multiple of $C_0$ on $C_1$ of $\r$, we may assume that $E_{13}$ is invertible. By adding suitable multiples of columns $n+1,...,l_1$ of $C_1$ on columns $1,...,n$ of $C_1$, we can assume that $E_{12}=0$ (Note that such local equivalence also changes $F_1,C_{22},\lambda_{j1}I_{l_1}$ for every $j>2$).
We conclude that except for $a_n$ and the last entry of $E_{10}$ all
other entries of the $n$th column of $C_1$ are 0. By subtracting
from $C_1$ a scalar multiple of $C_0$, we may assume that the last
entry of $E_{10}$ is 0. Now $a_n$ is the only nonzero entry in the
$n$th column of $C_1$.

We can choose an index $i>2$ such that $\lambda_{i1}\ne\lambda_{ik}$. By
replacing $C_i$ with $C_i-\lambda_{i1}C_0$ using local equivalence, the $n$th
column of $C_i$ becomes 0.
By Lemma \ref{le:2xn=NPT} (iv), we may assume that
the state
$[E_k,(\lambda_{ik}-\lambda_{i1})I_{l_k}]^\dg[E_k,(\lambda_{ik}-\lambda_{i1})I_{l_k}]$
of rank $l_k$ is PPT. Since its B-local rank is also $l_k$,
the matrix $E_k$ must be normal. As it is also lower triangular, it
must be a diagonal matrix. Because $\r$ is not 1-distillable, using Lemma \ref{le:kxk} we can further assume that
 \bea
 E_k = \m_{1} I_{n_1}\oplus\cdots\oplus \m_{s} I_{n_{s}}, \quad
 F_k = G_{n_1}\oplus\cdots\oplus G_{n_{s}};
 \quad n_1 + \cdots + n_{s} = l_k,
 \eea
with $G_j$ upper triangular of order $n_j$ for each $j$, $\m_i\ne \m_j$ and $s\ge1$. Then the
$R$th row of $C$ is a product vector, which shows that $\cR(\r)$ contains a product vector.
This is a contradiction with the hypothesis, and the proof is completed.  
\end{proof}

Unfortunately, we are not able to distill the $3\times N$ B-irreducible NPT state of rank $N+1$. If this case was proven, then Lemma \ref{le:reducible} and Theorem \ref{thm:MxNrank(N+1)NPT} would show that every $M\times N$ NPT state of rank $N+1$ is 1-distillable. Formally, we state it as follows.

\begin{lemma}
\label{le:equivalence}
Every $M\times N$ NPT state of rank $N+1$ is 1-distillable, if and only if every $3\times N$ B-irreducible NPT state of rank $N+1$ is 1-distillable.
\end{lemma}
\begin{proof}
It suffices to prove the "if" part. We assume that every $3\times N$ B-irreducible NPT state of rank $N+1$ is 1-distillable. Combined with Theorem \ref{thm:MxNrank(N+1)NPT} we know that every $M\times N$ B-irreducible NPT state of rank $N+1$ is 1-distillable. Using Lemma \ref{le:reducible}(ii), we obtain every $M\times N$ NPT state of rank $N+1$ is 1-distillable.
\end{proof}

\section{conclusion}
\label{sec:con}
We have shown that a bipartite NPT state of Schmidt three is 1-undistillable under LOCC. The distillablility of bipartite NPT states of Schmidt rank at least four remains an open problem. We also have shown that an $M\times N$ B-irreducible NPT state of rank $\max\{M,N\}+1$ is 1-distillable when $M,N>3$. An open problem from this paper is to distill $3\times N$ B-irreducible NPT states of rank $N+1$ or larger rank. The first case is $N=4$, because every rank-four NPT state is known to be 1-distillable. A possible approach for the problems is to consider its irreducibility and product states in its range and kernel.

\section*{Acknowledgments}

Authors were supported by the NNSF of China (Grant No. 11871089), and the Fundamental Research Funds for the Central Universities(Grant No. ZG216S2005).

\bibliographystyle{unsrt}

\bibliography{distill}

\begin{thebibliography}{10}

\bibitem{PhysRevA.59.4206}
Michal Horodecki and Pawel Horodecki.
\newblock Reduction criterion of separability and limits for a class of
  distillation protocols.
\newblock {\em Phys. Rev. A}, 59:4206--4216, Jun 1999.

\bibitem{PhysRevLett.69.2881}
Charles~H. Bennett and Stephen~J. Wiesner.
\newblock Communication via one- and two-particle operators on
  einstein-podolsky-rosen states.
\newblock {\em Phys. Rev. Lett.}, 69:2881--2884, Nov 1992.

\bibitem{teleportation1993}
Charles~H. Bennett, Gilles Brassard, Claude Cr\'epeau, Richard Jozsa, Asher
  Peres, and William~K. Wootters.
\newblock Teleporting an unknown quantum state via dual classical and
  einstein-podolsky-rosen channels.
\newblock {\em Phys. Rev. Lett.}, 70:1895--1899, Mar 1993.

\bibitem{PhysRevLett.67.661}
Artur~K. Ekert.
\newblock Quantum cryptography based on bell's theorem.
\newblock {\em Phys. Rev. Lett.}, 67:661--663, Aug 1991.

\bibitem{PhysRevLett.77.2818}
David Deutsch, Artur Ekert, Richard Jozsa, Chiara Macchiavello, Sandu Popescu,
  and Anna Sanpera.
\newblock Quantum privacy amplification and the security of quantum
  cryptography over noisy channels.
\newblock {\em Phys. Rev. Lett.}, 77:2818--2821, Sep 1996.

\bibitem{365700}
P.W. Shor.
\newblock Algorithms for quantum computation: discrete logarithms and
  factoring.
\newblock In {\em Proceedings 35th Annual Symposium on Foundations of Computer
  Science}, pages 124--134, 1994.

\bibitem{PhysRevA.54.3824}
Charles~H. Bennett, David~P. DiVincenzo, John~A. Smolin, and William~K.
  Wootters.
\newblock Mixed-state entanglement and quantum error correction.
\newblock {\em Phys. Rev. A}, 54:3824--3851, Nov 1996.

\bibitem{2012Distillability}
L.~Chen and Dragomir Okovi.
\newblock Distillability and ppt entanglement of low-rank quantum states.
\newblock {\em Journal of Physics A Mathematical \& Theoretical},
  44(28):1213--1219, 2012.

\bibitem{2000Evidence}
D.~P. Divincenzo, P.~W. Shor, J.~A. Smolin, B.~M. Terhal, and A.~V. Thapliyal.
\newblock Evidence for bound entangled states with negative partial transpose.
\newblock {\em Physical Review A}, 61(6):062312, 2000.

\bibitem{PhysRevA.61.062312}
David~P. DiVincenzo, Peter~W. Shor, John~A. Smolin, Barbara~M. Terhal, and
  Ashish~V. Thapliyal.
\newblock Evidence for bound entangled states with negative partial transpose.
\newblock {\em Phys. Rev. A}, 61:062312, May 2000.

\bibitem{horodecki2020open}
Pawel Horodecki, Lukasz Rudnicki, and Karol Zyczkowski.
\newblock Five open problems in quantum information theory.
\newblock {\em PRX Quantum}, 3:010101, Mar 2022.

\bibitem{Hiroshima2010APO}
Tohya Hiroshima.
\newblock A problem of existence of bound entangled states with non-positive
  partial transpose and the hilbert's 17th problem.
\newblock {\em arXiv: Quantum Physics}, 2010.

\bibitem{PhysRevA.68.022319}
Somshubhro Bandyopadhyay and Vwani Roychowdhury.
\newblock Classes of $n$-copy undistillable quantum states with negative
  partial transposition.
\newblock {\em Phys. Rev. A}, 68:022319, Aug 2003.

\bibitem{PhysRevA.78.022318}
Lin Chen and Yi-Xin Chen.
\newblock Rank-three bipartite entangled states are distillable.
\newblock {\em Phys. Rev. A}, 78:022318, Aug 2008.

\bibitem{Horodecki_2009}
Ryszard Horodecki, Pawel Horodecki, Michal Horodecki, and Karol Horodecki.
\newblock Quantum entanglement.
\newblock {\em Reviews of Modern Physics}, 81(2):865--942, Jun 2009.

\bibitem{Pankowski2007AFS}
Lukasz Pankowski, Marco Piani, Michal Horodecki, and Paweł Horodecki.
\newblock A few steps more towards npt bound entanglement.
\newblock {\em IEEE Transactions on Information Theory}, 56:4085--4100, 2007.

\bibitem{2003Rank}
P.~Horodecki, J.~A. Smolin, B.~M. Terhal, and A.~V. Thapliyal.
\newblock Rank two bipartite bound entangled states do not exist.
\newblock {\em Theoretical Computer Science}, 2003.

\bibitem{PhysRevLett.78.574}
Michal Horodecki, Pawel Horodecki, and Ryszard Horodecki.
\newblock Inseparable two spin- $\frac{1}{2}$ density matrices can be distilled
  to a singlet form.
\newblock {\em Phys. Rev. Lett.}, 78:574--577, Jan 1997.

\bibitem{2016Distillability}
L.~Chen and D.~Dokovic.
\newblock Distillability of non-positive-partial-transpose bipartite quantum
  states of rank four.
\newblock {\em Phys.rev.a}, 94(5):052318, 2016.

\bibitem{horodecki1997}
P.~Horodecki.
\newblock Separability criterion and inseparable mixed states with positive
  partial transposition.
\newblock {\em Phys. Lett. A}, 232:333, 1997.

\bibitem{PhysRevA.93.042331}
Lin Chen and Li~Yu.
\newblock Entanglement cost and entangling power of bipartite unitary and
  permutation operators.
\newblock {\em Phys. Rev. A}, 93:042331, Apr 2016.

\bibitem{Chen2016ProductSA}
Lin Chen and Li~Yu.
\newblock Product states and schmidt rank of mutually unbiased bases in
  dimension six.
\newblock {\em Journal of Physics A: Mathematical and Theoretical}, 50, 2016.

\bibitem{Song2023ProofOA}
Zhiwei Song, Lin Chen, Yize Sun, and Mengyao Hu.
\newblock Proof of a conjectured 0-r{\'e}nyi entropy inequality with
  applications to multipartite entanglement.
\newblock {\em IEEE Transactions on Information Theory}, 69:2385--2399, 2023.

\bibitem{2008On}
T.~Cubitt, A.~Montanaro, and A.~Winter.
\newblock On the dimension of subspaces with bounded schmidt rank.
\newblock {\em Journal of Mathematical Physics}, 49(2):95--179, 2008.

\bibitem{Cariello_2021}
Daniel Cariello.
\newblock Schmidt rank constraints in quantum information theory.
\newblock {\em Letters in Mathematical Physics}, 111(3), jun 2021.

\bibitem{BISWAS2023128610}
Indranil Biswas, Atanu Bhunia, Indrani Chattopadhyay, and Debasis Sarkar.
\newblock Entangled state distillation from single copy mixed states beyond
  locc.
\newblock {\em Physics Letters A}, 459:128610, 2023.

\bibitem{eps}
Masaki Owari, Samuel~L. Braunstein, Kae Nemoto, and Mio Murao.
\newblock $\epsilon$-convertibility of entangled states and extension of
  schmidt rank in infinite-dimensional systems.
\newblock {\em Quantum Info. Comput.}, 8(1):30–52, jan 2008.

\bibitem{BDM+99}
C.~H. {Bennett}, D.~P. {Divincenzo}, T.~{Mor}, P.~W. {Shor}, J.~A. {Smolin},
  and B.~M. {Terhal}.
\newblock {Unextendible Product Bases and Bound Entanglement}.
\newblock {\em Physical Review Letters}, 82:5385--5388, June 1999.

\bibitem{1995Purification}
Charles~H. Bennett, Gilles Brassard, Sandu Popescu, Benjamin Schumacher,
  John~A. Smolin, and William~K. Wootters.
\newblock Purification of noisy entanglement and faithful teleportation via
  noisy channels.
\newblock {\em Phys. Rev. Lett.}, 76:722--725, Jan 1996.

\bibitem{2000Operational}
Pawe\l{} Horodecki, Maciej Lewenstein, Guifr\'e Vidal, and Ignacio Cirac.
\newblock Operational criterion and constructive checks for the separability of
  low-rank density matrices.
\newblock {\em Phys. Rev. A}, 62:032310, Aug 2000.

\bibitem{2012Properties}
Lin Chen and Dragomirz Djokovic.
\newblock Properties and construction of extreme bipartite states having
  positive partial transpose.
\newblock {\em Communications in Mathematical Physics}, 323(1):241--284, Jul
  2013.

\bibitem{000266183600001}
Otfried Guehne and Geza Toth.
\newblock Entanglement detection.
\newblock {\em PHYSICS REPORTS-REVIEW SECTION OF PHYSICS LETTERS},
  474(1-6):1--75, APR 2009.

\bibitem{LHenderson_2001}
L~Henderson and V~Vedral.
\newblock Classical, quantum and total correlations.
\newblock {\em Journal of Physics A: Mathematical and General}, 34(35):6899,
  aug 2001.

\bibitem{2019Separability}
Gemma~De las Cuevas, Tom Drescher, and Tim Netzer.
\newblock Separability for mixed states with operator schmidt rank two.
\newblock {\em Quantum}, 3:203, dec 2019.

\bibitem{2000Separability}
B.~Kraus, J.~I. Cirac, S.~Karnas, and M.~Lewenstein.
\newblock Separability in 2×n composite quantum systems.
\newblock {\em Physical Review A}, 61(6):200--200, 2000.

\end{thebibliography}

\end{document}